\documentclass[]{gIPE2xxx}
\newcommand{\R}{\mathbb{R}}
\newcommand{\C}{\mathbb{C}}
\newcommand{\Z}{\mathbb{Z}}
\newcommand{\N}{\mathbb{N}}
\newcommand\ulx{\underline{x}}
\newcommand\ulf{\underline{F}}

\newcommand\ull{\underline{L}}
\newcommand\uly{\underline{y}}

\usepackage{graphicx}
\usepackage{epsfig}
\usepackage{caption}

\usepackage{color} 
\usepackage[normalem]{ulem} 

\newcommand{\FS}[1]{#1} 
\newcommand{\RS}[1]{\mathcal{#1}} 

\begin{document}


\markboth{D.~Gerth, B.~Hofmann, S.~Birkholz, S.~Koke, and
G.~Steinmeyer}{Regularization of an autoconvolution problem in
ultrashort laser pulse characterization}


\title{{\itshape Regularization of an autoconvolution problem in
ultrashort laser pulse characterization}}

\author{Daniel Gerth$^{\rm a,b}$,\, Bernd Hofmann$^{\rm a}$ $^{\ast}$
\thanks{$^\ast$Corresponding author. Email: hofmannb@mathematik.tu-chemnitz.de
\vspace{6pt}},\, Simon Birkholz$^{\rm c}$,\, Sebastian Koke$^{\rm
c}$,\, and~G\"unter~Steinmeyer$^{\rm c,d}$\\\vspace{6pt} $^{\rm
a}${\em{Chemnitz University of Technology, Department of
Mathematics, 09107 Chemnitz, Germany}}; $^{\rm b}${\em{Johannes
Kepler University, Industrial Mathematics Institute,
Altenbergerstra{\ss}e 69, 4040 Linz, Austria}}; $^{\rm c}${\em{Max
Born Institute for Nonlinear Optics and Short Pulse Spectroscopy,
12489 Berlin, Germany}}; $^{\rm d}${\em{Optoelectronics Research
Centre, Tampere University of Technology, 33101 Tampere,
Finland}}\\\vspace{6pt}\received{January 2013}}

\maketitle

\begin{abstract}
An ill-posed inverse problem of autoconvolution type is investigated. This inverse problem occurs in nonlinear optics in the context of ultrashort laser pulse characterization.
The novelty of the mathematical model consists in a physically required extension of the deautoconvolution problem beyond the classical case usually discussed in literature: (i) For measurements of ultrashort laser pulses with the self-diffraction SPIDER method, a stable approximate solution of an autocovolution equation with a complex-valued kernel function is needed. \linebreak (ii) The considered scenario requires complex functions both, in the solution as well as in the right-hand side of the integral equation. Since, however, noisy data are available not only for amplitude and phase functions of the right-hand side, but also for the amplitude of the solution, the stable approximate reconstruction of the associated smooth phase function represents the main goal of the paper.
An iterative regularization approach will be described that is specifically adapted to the physical situation in pulse characterization, using a non-standard stopping rule for the iteration process of computing regularized solutions. The opportunities and limitations of regularized solutions obtained by our approach are illustrated by means of several case studies for synthetic noisy data and physically realistic complex-valued kernel functions. Based on an example with focus on amplitude perturbations, we show that the autoconvolution equation is locally ill-posed everywhere. To date, the analytical treatment of the impact of noisy data on phase perturbations remains an open question. However, we show its influence with the help of numerical experiments. Moreover, we formulate assertions on the non-uniqueness of the complex-valued autoconvolution problem, at least for the simplified case of a constant kernel. The presented results and figures associated with case studies illustrate the ill-posedness
phenomena also for the case of non-trivial complex kernel functions.

\bigskip


\end{abstract}

\section{Introduction} \label{sec:intro}

In the early 1990s motivated by applications from spectroscopy (cf.~\cite{Bau91}) and stochastics (cf.,~e.g., \cite[p.74]{RI92})
contributions to the deeper mathematical and numerical analysis of deautoconvolution problems as a class of inverse problems in spaces of continuous or quadratically integrable real functions were made. Precisely, deautoconvolution problems under consideration were mostly aimed at finding non-negative functions $x$ with compact support ${\rm supp}(x) \subset \R$ from its autoconvolution $x*x$. After transformation to the unit interval ${\rm supp}(x) \subseteq [0,1]$ they consist in finding solutions of the autoconvolution equation
\begin{equation} \label{eq:auto1}
\int \limits _0^s x(s-q)\,x(q)\,d q\,= \, y(s)\,,
\end{equation}
where the support of $y$ belongs to the interval $[0,2]$. Since the autoconvolution operator $x \mapsto x*x$ is nonlinear and `smoothing', the deautoconvolution problem is ill-posed in the sense that for given $y$ the solutions $x$ need not be uniquely determined and mainly small
perturbations in the right-hand side $y$ caused by noisy data may lead to arbitrarily large errors in the solution. To overcome the negative consequences of ill-posedness up to some extent
some kind of regularization is required. Regularization techniques allow us to find stable approximate solutions of equation (\ref{eq:auto1}) based on auxiliary problems.
For data $y$ on the subinterval $0 \le s \le 1$ the paper \cite{GH94} has analyzed the situation of equation (\ref{eq:auto1}) and the application of Tikhonov's regularization method including its convergence properties, whereas in \cite{FH96} the situation of data $y$ on the whole interval $0\le s \le 2$ has been studied. Alternative regularization methods applied to equation (\ref{eq:auto1}) and specific numerical approaches were also discussed in \cite{Janno97,Janno00,FGH99,Ram02,CL05,vW08,DaiLamm08}.

Recently, the research group `Solid State Light Sources' of the Max Born Institute for Nonlinear Optics and Short Pulse Spectroscopy, Berlin, hit on the autoconvolution problem in the context of a new approach in ultrashort laser pulse characterization called Self-Diffraction SPIDER or short, SD-SPIDER (cf.~\cite{KBB08}). For phase reconstruction as an auxiliary problem the solution of equation (\ref{eq:auto1}) is needed, but for complex functions
$x: [0,1] \subset \R \to \C$ to be determined from complex observations $y: [0,2] \subset \R \to \C$. To our knowledge, a thorough analysis of the complex case in deautoconvolution is still
missing in the literature, in particular as the ill-posedness phenomenon arises in the complex case, too. Moreover, the occurrence of a device-related kernel function $k: [0,2]\times [0,1] \to \C$ involved in the mapping $x \mapsto y$, which is non-trivial in the sense that $k \not \equiv 1$, constitutes a challenging additional difficulty in connection with
this inverse problem behind SD SPIDER. So as a part of the SD SPIDER approach it would be necessary to solve (after transformation of the variables to the unit interval) the equation
\begin{equation} \label{eq:auto2}
\int \limits _0^s k(s,q)\; x(s-q)\,x(q)\,dq\,= \, y(s),\qquad 0 \le s \le 2,
\end{equation}
in a stable approximate manner when only noisy data of $y$ are given. To simplify the notation we write integrals like on the left-hand side of (\ref{eq:auto2})
even if a function in the integrand is not defined there, as this is the case for $x(s-q)$ if $s-q>1$. Then we set the corresponding function values as zero
and avoid to distinguish the integral representations for $0 \le s \le 1$ and $1< s \le 2$.

The equation (\ref{eq:auto2}) is a complex-valued and kernel-based generalization
of equation (\ref{eq:auto1}) with solution
\begin{equation} \label{eq:auto3}
x(q)\,=\,A(q)\,\exp[i \varphi(q)], \qquad 0 \le q \le 1,
\end{equation}
and right-hand side
\begin{equation} \label{eq:auto4}
y(s)\,=\,B(s)\,\exp[i \psi(s)], \qquad 0 \le s \le 2.
\end{equation}
We consider in this paper two different aspects of solving the integral equation (\ref{eq:auto2}) as a mathematical model for inverse problems. The general inverse problem (a)
consists in finding the complex function $x$ in  (\ref{eq:auto2})  from noisy data of the complex function $y$ and for a given complex-valued kernel $k$.
Alternatively, the SD-SPIDER-motivated specification (b) lies in finding
the phase function $\varphi$ in (\ref{eq:auto3}) from noisy data of $y$ and given $k$ when additional observations $\hat A$ of the amplitude function $A$ are available.
Below we will emphasize that the identifiability of the phase $\varphi$ requires measurement data of both the phase function $\psi$ as well as its corresponding amplitude function $B$ in (\ref{eq:auto4}).
In order to find the phase function $\varphi$ in problem (b) under some smoothness assumptions we will suggest a Tikhonov regularization approach (cf., e.g., \cite{EHN96,Yag98}) combined with a specifically adapted strategy for choosing the regularization parameter $\alpha>0$. In this context, iterative procedures of Levenberg-Marquardt type (cf., e.g., \cite{KNS08}) are required to compute acceptable approximations of the Tikhonov-regularized solutions in an efficient manner. Exploiting the additional knowledge of noisy data for the amplitude function $A$ allowed us to construct an adapted stopping rule for the developed iterative regularization approach. Moreover, we illustrate the quite acceptable work of this approach by a numerical case study.

Currently, not all questions concerning
(a) and (b) can be answered by the authors. Therefore, it makes sense to show the local ill-posedness of problem (a) by an analytic example and to illustrate ill-posedness phenomena of (b) by numerical case studies. At the moment the analytical treatment of
the impact of noisy data on phase perturbations is an open question. Moreover, it seems to be reasonable to present a uniqueness assertion on the inverse problem (a) in the case $k \equiv 1$ since assertions for general
complex kernels are missing to our knowledge. On the other hand, all numerical case studies for problem (b) were performed with physically relevant kernel functions $k$.

At this point we should note that the problem of finding a function (\ref{eq:auto3}) solving the equation (\ref{eq:auto2}) is a generalization of the problem of recovery of a compactly supported and complex-valued function from the modulus of its Fourier transform. This so-called phase retrieval problem and its applications, for example in optics, electron microscopy and astronomy, were intensively studied in the literature based on the seminal paper
\cite{Kli85}, and we refer to the review paper \cite{KST95}, to  \cite[Section~1.2]{SKHK12}, and references therein. In particular, rigorous uniqueness results were proven for the one- and two-dimensional phase retrieval problem
(cf., e.g., \cite{Kli85,Kli89,KS94,Kli06}). Moreover, numerical methods were proposed in \cite{KS94,KST95}.

The paper is organized as follows: in Section 2 we briefly review the role of deconvolution and decorrelation in the characterization of laser pulses. Then we analyze in some more detail the previously unsolved deconvolution problem in self-diffraction SPIDER, i.e., a specific variant of laser pulse characterization methods. Leading from the notation commonly used in the physical literature to the one employed in mathematics, we reformulate the abstract mathematical problem behind and investigate its ill-posedness in Section 3. Subsequently, in Section 4, we derive consequences of the Titchmarsh convolution theorem for a constant kernel function $k$ and present an adapted regularization approach in Section 5. Finally, we consider several case studies based on synthetic noisy data and illustrate the convergence behavior of the iteration procedure, which requires a non-standard stopping rule. After a summary of the autoconvolution problem and its specific regularization
approach, we conclude with a brief outlook at the relevance of the findings in the physical sciences.

\section{Physical background} \label{sec:physics}

\subsection{The evolution of ultrashort pulse characterization}

Ultrashort laser pulses constitute the shortest man-made
controllable events, with demonstrated pulse durations of 4
femtoseconds ($4 \times 10^{-15}$\,seconds) and below
\cite{SS06,RBH08,KLH10}. Using wavelength conversion techniques,
pulses as short as $8 \times 10^{-17}$\,s have been produced
\cite{SBC06,GSH08}, which is in the range of the fastest transient
events in atoms and molecules. Quite remarkably, these generated
pulses approach a limiting width of a single optical cycle of the
underlying electric field carrier, as illustrated in Figure~\ref{fig:IllustrateTimeFrequency}. Nevertheless, these remarkable
technological achievements also cause a serious dilemma for their
accurate characterization and measurement. Temporal resolution of a
dynamical process is always limited by the gate time of the
sampling process, similar as in photography, where temporal resolution is
dictated by fastest available shutter speed. As there simply exists
no gate process faster than the duration of the shortest laser
pulses, the history of laser pulse characterization and ultrafast
spectroscopy has also always been a history of deconvolution and
decorrelation.

\begin{figure}
\includegraphics[width=\textwidth]{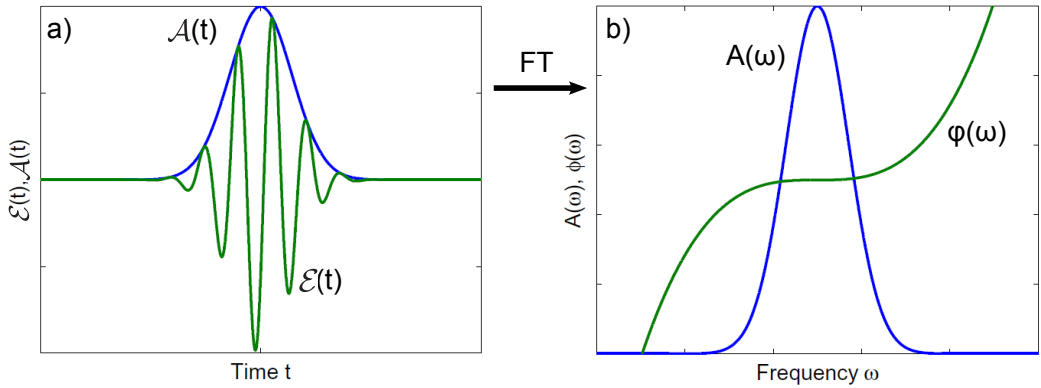}
\caption{(a, left) Illustration of the pulse to be characterized. The shortest generated laser pulses comprise only few oscillation of the electrical field $\RS{E}(t)$. The aim of pulse characterization is to retrieve the temporal evolution of the envelope of the electric field $\RS{A}(t)$ or of the intensity $\RS{I}(t) = |\RS{A}(t)|^{2}$. (b, right) The Fourier domain representation of the pulse is often utilized within the process of pulse retrieval. Therefore, it is decomposed into its amplitude $\FS{A}(\omega)$ and the spectral phase $\FS{\varphi}(\omega)$.}\label{fig:IllustrateTimeFrequency}
\end{figure}

From a physicist's point of view, a practical and easy-to-implement remedy for the fundamental dilemma of pulse characterization is the use of identical replicas of the input pulse as the pulse under investigation and for the gate function \cite{A67}. For
implementation of the temporal gate, one has to form the product of
the two functions, which is conveniently done by employing a
nonlinear optical process. In the simplest case of sum frequency
generation, this process generates the product of two intensity
envelopes $\RS{I}_1(t)$ and $\RS{I}_2(t)$, with one of these envelopes
serving as the gate event. Varying the delay between these
temporally dependent signals then allows recording their
cross-correlation
\begin{equation}
\RS{C}(\tau) \propto \int_{-\infty}^{\infty} \RS{I}_1(t) \RS{I}_2(t-\tau) \, {\rm
d}t,
 \label{eq:corr}
\end{equation}
where we employed the usual sign convention for the delay $\tau$ in
optics and the integration over time $t$ resembles the integration that is carried out by the slow detector recording the correlation signal. Reconstruction of $\RS{I}_1(t)$ with known correlation function
$\RS{C}(\tau)$ and gate function $\RS{I}_2(t)$ constitutes an inverse problem, which is frequently solvable with satisfactory
precision, see, e.g., \cite{GSK01} for a practical example. It is
important to understand that experimental noise mandates the gate
function to be kept as simply structured as possible. Ideally, it
should be chosen a single-maximum function with smallest possible
temporal duration.

Using identical functions $\RS{I}(t) \equiv \RS{I}_1(t)= \RS{I}_2(t)$,
Eq.~(\ref{eq:corr}) describes an autocorrelation, and the
decorrelation problem turns out to be only ambiguously solvable
\cite{H64,BS79}. This can be easily seen by use of the Fourier
convolution theorem. In the spectral domain, the correlation is
written as a product
\begin{equation}
\FS{C}(\omega) \propto
\FS{I}(\omega)\FS{I}^\ast(\omega)=\left\vert \FS{I}(\omega)
\right\vert^2, \label{eq:fourier}
\end{equation}
where $\omega$ is the optical angular frequency and $\FS{I}(\omega)$
is the spectral intensity. The Fourier
transform $\FS{f}(\omega)$ of a function $\RS{F}(t)$ is given by
\begin{equation}
\FS{f}(\omega)=\int\limits_{-\infty}^{\infty} \RS{F}(t) \exp\left(-
i \omega t \right) {\rm d}t.
\end{equation}
In Eq.~(\ref{eq:fourier}), the phase of the complex-valued
$\FS{I}(\omega)$ does not affect the correlation signal
$\FS{C}(\omega)$. This means that all information on temporal pulse
asymmetry in the time domain does not enter into the autocorrelation
signal. Simplistic reconstruction, trying to revert the absolute
square in (\ref{eq:fourier}) by the principal square root,
therefore always yields a symmetric reconstructed $\RS{I}(t)$.
Consequently, decorrelation does not unambiguously work unless
additional information on pulse asymmetry is obtained.

Several ways have been investigated to resolve this issue. Early
attempts \cite{NMY89,PR98} relied on the availability of additional
spectral information $\FS{I}(\omega)$, further attempts included
more complex types of autocorrelations \cite{HLN07}. It could be
shown, however, that even the tiniest amounts of measurement noise
render all these concepts unpractical \cite{CW01}, yet again
thwarting unambiguous reconstruction of $\RS{I}(t)$. In all these types
of decorrelation problems, therefore, thorough analysis typically
pointed out severe limitations rather than offering practical
solutions.

This unsatisfactory situation changed dramatically with methods like
frequency-resolved optical gating (FROG \cite{T00}) and the sonogram
technique \cite{R99}. In simple words, the idea of these method is
to replace the separate acquisition of a spectrum and an
autocorrelation by measurement of a two-dimensional spectrogram-like
function $\FS{S}(\tau,\omega)$ with simultaneous functional dependence on
delay and frequency. In the simplest case of second-harmonic FROG,
\begin{equation}
\FS{S}(\tau,\omega) \propto \int\limits_{-\infty}^{\infty} \RS{I}(t) \RS{I}(t-\tau)
\exp\left(-i \omega t \right) {\rm d}t . \label{eq:FROG}
\end{equation}
In the experiment, the acquisition of FROG traces only requires
replacing the spectrally integrating detector by a spectrally
resolving one and can be used with all existing autocorrelation
geometries. It can be shown \cite{T00} that FROG resolves most but
not all issues of decorrelation. There is a one-to-one
correspondence between a given FROG trace $\FS{S}(\tau,\omega)$ and the
intensity envelope $\RS{I}(t)$, provided that the pulse does not consist
of temporally or spectrally separated segments with intensity levels
approaching the experimental noise floor in between. A detailed
discussion of remaining ambiguities can be found in
\cite{KTO03,SST04}. One additional ambiguity of a second-order
autocorrelation-type FROG method is time reversal, i.e., the method
cannot distinguish between $\RS{I}(t)$ and $\RS{I}(-t)$. Nevertheless, there are
alternative methods relying on third-order nonlinearities \cite{T00}
and other fixes \cite{ZSK02} to overcome this issue. Retrieving
$\RS{I}(t)$ from a known $\FS{S}(\tau,\omega)$ is yet again an inverse problem, and considerable effort went into
reliable and numerically efficient algorithms for retrieving the
pulse shape from a FROG trace
\cite{DFR94,DFT98,T00,ACL05,LT08,SS09}.

As an alternative to FROG, spectral phase interferometry for direct
electric-field reconstruction (SPIDER) emerged \cite{IW98,SS06a}.
This method is conceptually different from all previous approaches
and normally does not require any kind of deconvolution. Being based
on spectral interferometry \cite{RSB89,FBS96}, SPIDER directly
measures the spectral phase $\FS{\varphi}(\omega)$ of the
electric field $\FS{E}(\omega)=\FS{A}(\omega) \exp\left[i
\FS{\varphi}(\omega)\right]$, where both the spectrum $\FS{A}(\omega)$
and the phase $\FS{\varphi}(\omega)$ are real-valued functions.
In the time domain, the field can again be separated
into amplitude and phase via $\RS{E}(t)=\RS{A}(t)\exp\left[i \phi(t)\right]$. The
electric field relates to the intensity via $\RS{I}(t)=\vert \RS{A}(t)\vert^2$.
SPIDER requires the generation of two
replicas of the input pulse $\RS{I}(t)$, which are processed with an
ancillary pulse $\RS{I}_{\rm a}(t)$ in a second-order nonlinear crystal. This
ancillary pulse can also be derived from the same input field, i.e.,
the method is self-referenced. To this end, a fraction of the input
pulse is sent through a dispersive medium, e.g., a glass block. The
dispersion of the glass block gives rise to a non-vanishing chirp
$\zeta(t)={\rm d}\phi(t)/{\rm d}t$, i.e., the pulse's carrier
frequency varies with time. Mixing the field of the ancillary pulse
$\RS{E}_{\rm a}$ with the two replicas contained in $\RS{E}_{\rm r}$ causes a
different frequency shift for each replica, which is referred to as
spectral shear. The spectrum of the nonlinear mixing product, i.e.,
the Fourier transform of $\RS{E}_{\rm r}(t) \RS{E}_{\rm a}(t)$ can then be
analyzed with the methods described in \cite{TIK82,BS06} to
reconstruct the spectral phase $\FS{\varphi}(\omega)$ of the
original input pulse. With additional knowledge of the amplitude
$\FS{A}(\omega)$, a simple Fourier transform then suffices for
complete reconstruction of $\RS{I}(t)$ in the time-domain.

While traditional SPIDER methods all rely on second-order
nonlinearities, it seems appealing to extend this proven method to
third-order nonlinearities as they are frequently used for
autocorrelation and FROG. These higher-order methods allow generating a
FROG or SPIDER signal that is spectrally collocated with the
generating wave, which alleviates pulse characterization in the
ultraviolet \cite{KBB08}. Moreover, third-order SPIDER has recently
been demonstrated using a monolithic waveguide device \cite{PPP11},
bringing us closer to the dream of an all-optical oscilloscope on a
single optoelectronic chip. As third-order optical nonlinearities
appear in any kind of dielectric material and not only in
non-isotropic crystalline media, their use is essential for
optically integrated characterization methods. Unfortunately FROG, autocorrelation
and also a SD-based spectral-interference pulse characterization technique~\cite{Liu12}
involve mechanical scanning of optical
delays, which rules them out for integrated optical devices. At the
current state of the art, this only leaves SPIDER for such
applications.

A third-order SPIDER, however, turns out to be difficult to
implement \cite{KBB08}. Given that there is now the product
formation of three waves involved, which interact in the nonlinear
mixing process, one can either generate two waves from the ancillary
pulse and one from the two replicas, or, vice versa, two replica
waves $\RS{E}_{\rm r}(t)$ and only one ancilla $\RS{E}_{\rm a}$. As the
chirped and temporally stretched ancilla typically displays a much
lower peak intensity than the two replicas, the latter approach with
two replica waves and one ancilla therefore constitutes the more
favorable physical scenario, resulting in higher conversion
efficiencies. Using, for example, a SD geometry
\cite{VKO79}, the nonlinear mixing process forms a signal
$\RS{E}_{\text{SD}}(t) \propto \RS{E}^2_{\rm r}(t) \RS{E}^\ast_{\rm a}(t)$. This duplicates
the scenario of second-order SPIDER with two notable differences.
The replicas enter squared $\RS{E}^2_{\rm r}(t)$, and the ancilla enters
complex conjugated. While the latter essentially only corresponds to
a sign reversal of $\zeta(t)$, the squared replicas cause SPIDER to
measure the spectral phase $\varphi_{\rm conv}(\omega)$ of the convolved signal $\RS{E}_{\rm conv}(t) \propto \RS{E}^2_{\rm r}(t)$ rather than the spectral phase of the pulse itself. Hence, a deconvolution is required in order to gain insight about the shape of the electric field $\RS{E}(t)$ or of the intensity envelope $\RS{I}(t)$.

\subsection{The deconvolution problem in SD SPIDER}

The derivation of the deconvolution task formulated in Eq.~(2) -- (4) requires a more detailed analysis of the generation process of the SD wave. In the last paragraph, a simplifying time-domain-based formulation of the nonlinear signals generated for pulse characterization has been employed in order to work out the basic differences between the different characterization methods. A more general formulation is obtained in the Fourier domain and by the use of the nonlinear wave equation for the involved electro-magnetic waves ~\cite{Shen2003}. From this approach  one can derive under the slowly varying envelope approximation  that the strength of the SD signal is given by

\begin{equation} \label{eq:spiderintL}
\FS{E}_{\text{SD}}(L,\omega) = i \frac{\mu_0 c \omega}{2 n} \int_0^L \FS{P}_{\rm nl}(\omega) e^{-i \kappa_{\text{SD}}(\omega) z}\, d z \,,
\end{equation}
where $\FS{P}_{\rm nl}(z, \omega)$ is the nonlinear polarization of the bound electrons inside the materials employed for generation of the SD signal. The vacuum permeability $\mu_0$, the speed of light $c$ and the linear refractive index $n=n(\omega)$ are physical constants. $z$ denotes propagation direction of the electro-magnetic waves inside the material that is used to generate the SPIDER signal and $L$ denotes the material's length. $\omega$ is the frequency of the self-diffraction signal and $\kappa_{\text{SD}}(\omega)$ is the wave number of this signal. The constants in Eq.~(\ref{eq:spiderintL}) and the integral over the interaction length will contribute to the kernel $k$ of the autoconvolution problem Eq.~(\ref{eq:auto2}). For further analysis the nonlinear polarization $\FS{P}_{\rm nl}(\omega)$ has to be considered, which is in general given by
\begin{equation} \label{eq:spiderfreq}
    \FS{P}_{\rm nl}(\omega) = \int_0^\infty \int_0^\infty \chi^{(3)}(\omega, -\omega_{\rm a}, \omega_1, \omega_2)
    \FS{E}_{\rm a}^\ast \FS{E}_{\rm r}(\omega_1) \FS{E}_{\rm r}(\omega_2) \delta(\omega + \omega_{\rm a} - \omega_1 - \omega_2) \, d\omega_2 d\omega_1 \,.
\end{equation}
Here, $\omega_1$, $\omega_2$ are the frequency components of the two replica waves $\FS{E}_{\rm r}$, and $\omega_{\rm a}$ denotes the single frequency component of the ancilla $\FS{E}_{\rm a}$. The values of $\chi^{(3)}$ and $\delta(\omega + \omega_{\rm a} - \omega_1 - \omega_2)$ express physical constraints and relations of the light-matter-interaction:
$\chi^{(3)}$, the third-order nonlinear susceptibility, denotes the interaction strength between the field of light and the material in use. The factor \linebreak $\delta(\omega + \omega_{\rm a} - \omega_1 - \omega_2)$ ensures energy conservation in the expression above. Inserting of Eq.~(\ref{eq:spiderfreq}) in Eq.~(\ref{eq:spiderintL}) and simplification of the result yields
\begin{equation} \label{eq:spiderfreqsimple}
\begin{split}
    \FS{E}_{\text{SD}}(\omega) = &  \int_0^{\omega + \omega_{\rm a}} \chi^{(3)}(\omega, -\omega_{\rm a}, \omega_1, \omega+\omega_{\rm a}-\omega_1)  \, M (\omega,-\omega_{\rm a}, \omega_1,\omega_2) \\
  & \phantom{\int_0^{\omega + \omega_{\rm a}}}\, \FS{E}_{\rm a}^\ast \FS{E}_{\rm r}(\omega_1) \FS{E}_{\rm r}(\omega + \omega_{\rm a} - \omega_1) \, d\omega_1 \, .
\end{split}
\end{equation}
Here, $M$ collects the constants
and the result of the $z$-integration of the exponential term in Eq.~(\ref{eq:spiderintL}) describing the phase matching conditions between all waves involved in this four-wave-mixing process:
The frequency dependent values of $\chi^{(3)}$  and $M$ are known from experiments and merge to the kernel
\begin{equation}\label{eq:spiderkernel}
\FS{K}(\omega,\omega_1,\omega_{\rm a})=\frac{\mu_0 c \omega}{2 n}\chi^{(3)}(\omega, -\omega_{\rm a}, \omega_1, \omega+\omega_{\rm a}-\omega_1)\FS{E}_{\rm a}^\ast e^{i \Delta\kappa\cdot(\xi,\eta,\frac{L}{2})^T}\text{sinc}(\Delta\kappa_z L/2)
\end{equation}
where $\Delta\kappa=(\Delta\kappa_\xi, \Delta\kappa_\eta,\Delta\kappa_z)(\omega, -\omega_{\rm a}, \omega_1, \omega+\omega_{\rm a}-\omega_1)$ is the phase-mismatch vector and $(\xi,\eta,z)$ denotes spatial coordinates.
Now (\ref{eq:spiderfreqsimple}) can be written as
\begin{equation}\label{eq:spiderfreqkernel}
\FS{E}_{\text{SD}}(\omega) = \int_0^{\omega + \omega_{\rm a}} \FS{K}(\omega,\omega_1,\omega_{\rm a})\FS{E}_{\rm r}(\omega_1) \FS{E}_{\rm r}(\omega + \omega_{\rm a} - \omega_1)\, d\omega_1.
\end{equation}
We note that the wave $\FS{E}_{\rm r}$ has a spectrum of finite width, i.e., there exist frequencies $0<\omega_1^{\textup{min}}<\omega_1^{\textup{max}}<\infty$ such that $\FS{E}_{\rm r}(\omega_1)=0$ for $\omega_1\notin[\omega_1^{\textup{min}},\omega_1^{\textup{max}}]$. Consequently $\FS{E}_{\text{SD}}$ is also compactly supported and $\FS{E}_{\text{SD}}(\omega)=0$ for $\omega\notin[2\omega_1^{\textup{min}}-\omega_{\rm a},2\omega_1^{\textup{max}}-\omega_{\rm a}]$. Substituting
\begin{eqnarray}
\omega+\omega_{\rm a}&=s\, \omega_1^{\textup{max}}&\textup{ for } s\in[0,2],\nonumber\\
\omega_1&=q \,\omega_1^{\textup{max}}&\textup{ for } q\in[0,1]\nonumber
\end{eqnarray}
in Eq.~(\ref{eq:spiderfreqkernel}) and defining
\begin{equation} \label{eq:def}
\begin{split}
    k(s,q) :=&  \omega_1^{\textup{max}}\,\FS{K}(s\, \omega_1^{\textup{max}}  -\omega_{\rm a},
    q\,  \omega_1^{\textup{max}},\omega_{\rm a}), \\
    x(q) :=& \FS{E}_{\rm r}(q\, \omega_1^{\textup{max}}), \\
    y(s) :=& \FS{E}_{\text{SD}}(s\, \omega_1^{\textup{max}} -\omega_{\rm a}),
\end{split}
\end{equation}
allows us to trasnform the limits of the integral such that we arrive at the abstract mathematical model equation (\ref{eq:auto2}), which will be studied in the following. By designing the SPIDER apparatus in such a way that one of the replica beams can be blocked, the SD signal's amplitude $B(s)$ and its phase $\psi(s)$ in (\ref{eq:auto4}) can be measured with the same device such that
(noisy) data of the complex function $y(s),\;0 \le s \le 2$, are available.
An independent measurement of the spectral intensity $\FS{I}(\omega)=|A(\omega)|^{2}$ in front of the SPIDER apparatus allows us to compute
 approximations $\hat A(q)$ of $\FS{A}(q)$ corresponding to the incident pulse. Since for the retrieval of the pulse shape,  the complex function $ x(q) $, $0 \leq q \leq 1$ has to be Fourier transformed, the completely unknown phase function $\varphi(q) $ or its derivative $\rm{GD}(q)=\varphi'(q)$ called group delay remains to be determined.
For simplicity, in the following we neglect the uncertainty of
kernel data and suppose to know in a precise manner the \emph{continuous} complex function $k(s,q),\;(s,q) \in [0,2]\times [0,1]$. We denote by $$k_{\rm max}:=\max \limits_{(s,q) \in [0,2]\times [0,1]} |k(s,q)|$$ the maximum of its modulus.

\section{The abstract mathematical model and its ill-posedness}

The physical inverse problem under consideration described in Section~\ref{sec:physics} requires
the solution of the generalized autoconvolution equation (\ref{eq:auto2}).
As outlined in Section~\ref{sec:intro}, the two aspects (a) and (b) are under consideration. As is well-known the comprehension of additional information about expected solutions and the retrieval of data plays an important role for
the stable approximate solution of inverse problems. Therefore, physicists  are preferably interested in aspect (b) aimed at finding the continuously differentiable phase function  $\varphi: [0,1] \to \R$
in (\ref{eq:auto3}), given the kernel $k$ and noisy observational data of $y$ and $A$. We note that the additional knowledge of an estimate $\hat A(q), \;0 \le q \le 1,$ of the modulus function $|x(q)|=A(q), \;0 \le q \le 1$, in (\ref{eq:auto3}) will play a prominent role for choosing the regularization parameter in the process of constructing stable approximate solutions to (\ref{eq:opeq}) in Section~\ref{sec:adapreg}.
Nevertheless, it will give some insight into the problem structure to focus  in this and in the subsequent section on aspect (a), where
(\ref{eq:auto2}) is considered  as a nonlinear operator equation
\begin{equation} \label{eq:opeq}
F(x)\,=\,y\,,
\end{equation}
formulated in appropriate abstract function spaces.

Taking into account that the forward operator attains the form
\begin{equation} \label{eq:forward}
[F(x)](s):=\int \limits _0^s k(s,q)\; x(s-q)\,x(q)\,dq,\qquad 0 \le s \le 2,
\end{equation}
we are searching for the function  $x: [0,1] \to \C$ from noisy data $y^\delta \in L^2_{\C}(0,2)$ of $y$ that satisfy the deterministic noise model
\begin{equation} \label{eq:noise}
\|y^\delta-y\|_{L^2_{\C}(0,2)} \le \delta\,,
\end{equation}
with noise level $\delta>0$.
Our focus is on the Hilbert space situation $F:\,L^2_{\C}(0,1) \to L^2_{\C}(0,2),$ where $F$ is mapping
between Hilbert spaces of square-integrable complex functions. For $x:[0,1] \to \C,\;y \in [0,2] \to \C$, and $k: [0,2]\times [0,1] \to \C$ formula (\ref{eq:forward}) is an abbreviation of the
form
$$[F(x)](s):=\begin{cases} \quad \int \limits _0^s \;\;k(s,q)\; x(s-q)\,x(q)\,dq,\qquad 0 \le s \le 1,\\ \quad \int \limits _{s-1}^1 k(s,q)\; x(s-q)\,x(q)\,dq,\qquad 1 < s \le 2,  \end{cases}$$
which we always use for simplicity assuming that $x,y$ can be extended to $\R$ and $k$ to $\R^2$ as zero outside of the original domains.
In this context, we remember the structure of the norm
$\|z\|_{L^2_{\C}(0,a)}=\left(\int_0^a |z(q)|^2 dq\right)^{1/2}$ for elements $z \in  L^2_{\C}(0,a)$. If we consider
square-integrable real functions $z \in L^2(0,a)$, then the norm is the same, but $|z(q)|$ denotes the modulus of the real value $z(q)$.

The ill-posedness phenomena of non-uniqueness and instability well-established for the equation (\ref{eq:auto1}) (cf.~\cite{GH94,FH96}) also occur when the
complex-valued function $x: [0,1] \to \C$ is determined from Eq.~(\ref{eq:auto2}) and also in case that only $\varphi$ has to be found. Together with $x$ also the function $-x$ having the same modulus $|x|$ satisfies (\ref{eq:auto2}). Hence, if a continuously differentiable function $\varphi: [0,1] \to \R$ as a part of $x$ in (\ref{eq:auto3}) can be chosen such that equation (\ref{eq:auto2}) is satisfied for given functions $A$ and $y$, then $\varphi+\pi$ also solves the equation.   On the other hand, together with $\varphi$ also $\varphi+2\pi$ and hence all $\varphi+m \pi,\;m \in \Z,$ solve the equation. This type of non-uniqueness cannot be neglected, but is of inferior significance, since the corresponding group delays $\varphi^\prime$ are uniform for all integers $m$.

Instability is a more important difficulty occurring in all linear and nonlinear infinite dimensional least-squares problems which are aimed at solving inverse problems with smoothing forward operators.
Therefore, extreme care must be exercised when discretizing an infinite dimensional least-squares problem since the finite-dimensional approximating least-squares solutions may not converge (cf., e.g., \cite{Sei80})
or, even worse, they may diverge from the true solution with arbitrarily high speed (cf., e.g., \cite{ST06}). For the nonlinear operator $F$ from (\ref{eq:forward}) we unfortunately have that for every $x_0 \in  L^2_{\C}(0,1) $  there exist, in case of arbitrarily small radii $r>0$, sequences $\{x_n\}_{n=1}^\infty \subset  L^2_{\C}(0,1)\cap B_r(x_0)$ in a neighborhood  $B_r(x_0):=\{x \in L^2_{\C}(0,1):\,\|x-x_0\|_{L^2_{\C}(0,1)}<r\}$  of $x_0$  with
\begin{equation} \label{eq:localill}
\|F(x_n)-F(x_0)\|_{L^2_{\C}(0,2)} \to 0, \quad \mbox {but} \quad \|x_n-x_0\|_{L^2_{\C}(0,1)} \not \to 0  \quad \mbox{as} \quad n \to \infty.
\end{equation}
We call this local ill-posedness at the point $x_0$ and it has the consequence that a solution of equation (\ref{eq:auto2})
cannot be approximated arbitrarily good even if the noise level of the data tends to zero (cf.~\cite{HS94,HS98,HO00}). For injective operators $F$ local ill-posedness indicates that $F$ is not continuously
invertible, and this mostly results from compactness of the forward operator $F$. However, in \cite{GH94} it was shown that the autoconvolution operator from equation (\ref{eq:auto1}) is locally ill-posed everywhere
 in $L^2(0,1)$, but fails to be compact. So compactness also cannot be expected as an intrinsic property of $F$ in the complex case, but following an idea from \cite{FH96}, which was extended in \cite{Gerth11}, we can
nevertheless prove local ill-posedness (\ref{eq:localill}) of $F$ from (\ref{eq:auto2}) everywhere in $L^2_{\C}(0,1)$ by Example~\ref{ex:localill}.

\begin{example} \label{ex:localill}
For any radius $r>0$ the positive function $\Psi_\beta(q):=\frac{r\sqrt{1-2\beta}}{q^\beta},$ \linebreak $0 < q \le 1$, possesses  for all $0<\beta<\frac{1}{2}$ the properties $$\Psi_\beta \in L^2(0,1) \subset L^2_{\C}(0,1)\quad \mbox{and}\quad \|\Psi_\beta\|_{L^2_{\C}(0,1)}=r.$$ Then the equalities
$$[\Psi_\beta * \Psi_\beta](s)=r^2(1-2\beta)s^{1-2\beta}\int_0^1 (1-u)^{-\beta}u^{-\beta}du=r^2(1-2\beta)s^{1-2\beta}B(1-\beta,1-\beta)$$
with Euler's beta function $B(\cdot,\cdot)$, satisfying for $0<\beta<1/2$ the condition \linebreak $0<B(1-\beta,1-\beta) < \pi$,  yield the estimate
$$\|\Psi_\beta * \Psi_\beta\|_{L^2_{\C}(0,2)} \le \sqrt{2} \max \limits_{s \in [0,2]} [\Psi_\beta * \Psi_\beta](s) \le \sqrt{2}r^2 (1-2\beta)\pi 2^{1-2\beta} \to 0 \;\mbox{as} \; \beta \to \frac{1}{2}.$$
By setting $x_n:=x_0+ \Psi_{\frac{1}{2}-\frac{1}{n}}$ we have  for
every $x_0 \in L^2_{\C}(0,1)$ and all $0 \le s \le 2$
\begin{eqnarray*}
\begin{split}
|[F(x_n)](s)-[F(x_0)](s)| \le  \left|\int \limits_0^s
(k(s,q)+k(s,s-q))x_0(s-q) \Psi_{\frac{1}{2}-\frac{1}{n}}(q)dq \right| \\ +
\left|\int \limits_0^s k(s,q) \Psi_{\frac{1}{2}-\frac{1}{n}}(s-q)
\Psi_{\frac{1}{2}-\frac{1}{n}}(q) dq \right| \\\le
k_{\textup{max}}\{2\int \limits_0^s |x_0(s-q)| \Psi_{\frac{1}{2}-\frac{1}{n}}(q)dq+
[\Psi_{\frac{1}{2}-\frac{1}{n}}*\Psi_{\frac{1}{2}-\frac{1}{n}}](s)\} \to 0 \quad \mbox{as} \quad n \to \infty,
\end{split}
\end{eqnarray*}
because $\Psi_{\frac{1}{2}-\frac{1}{n}}$ is weakly convergent in $L^2(0,1)$ to the zero function as $n \to \infty$.
Since moreover $|[F(x_n)](s)-[F(x_0)](s)|$ is bounded from above by a constant that does not depend on $s$, Lebesgue's
dominated convergence theorem implies that $\|F(x_n)-F(x_0)\|_{L^2_{\C}(0,2)} \to 0$ as $n \to \infty$ and hence, with $\|x_n-x_0\|_{L^2_{\C}(0,1)}=r>0$, (\ref{eq:localill}) is valid. Consequently, we have local ill-posedness for all functions  $x_0 \in L^2_{\C}(0,1)$.
\end{example}

Note that with respect to $x(q)\,=\,A(q)\,\exp[i \varphi(q)]$ the local ill-posedness mentioned in Example~\ref{ex:localill} refers to exploding amplitudes $A(q)$ for small $q$ as a consequence of
the fact that $\Psi_{\beta}(q)$ has a weak pole at $q=0$. If $A(q)$ is given as a continuous function, then the ill-posedness with respect to the remaining unknown continuous phase function $\varphi$
is less obvious, but case studies show that instability also occurs in the sense that clearly distinguished phase functions $\varphi$ can lead to nearly the same complex-valued function $[F(x)](s)=B(s)\,\exp[i \psi(s)]$. Such studies also prove that $\varphi$ is not identifiable alone from the phase function $\psi$. Both functions, amplitude $B$ and phase $\psi$, are required to recover $\varphi$ when $A$ is known.

\section{Titchmarsh convolution theorem and its consequence}

For any function $x \in L^1_{\C}(0,1)$ the elements $F(x)$ and $F(-x)$ according to the operator $F$ from (\ref{eq:forward}) coincide and belong to $L^1_{\C}(0,2)$.
However, it is of interest whether this is the only ambiguity of solutions to equation (\ref{eq:auto2}). For the special case $k \equiv 1$ a positive answer can be given by the Titchmarsh convolution theorem (cf.~\cite{Tit26}) which we formulate as a lemma:

\begin{lemma} \label{lem:tit}
Let $f,g \in L^1_{\C}(\R)$ with ${\rm supp}(f) \subset [0,\infty),\; {\rm supp}(g) \subset [0,\infty)$, and let for some constant $a>0$
$$[f*g](s):=\int \limits_0^s  f(s-q)\, g(q)\, dq \,=\,0 \qquad \mbox{for almost all}\quad s \in [0,a]\,.  $$
Then there are non-negative constants $a_1$ and $a_2$ such that $a_1+a_2 \ge a$ and
$$f(q)=0\quad \mbox{for almost all}\quad t \in [0,a_1], \quad g(q)=0\quad \mbox{for almost all}\quad q \in [0,a_2].$$
\end{lemma}

This lemma allows us to prove the following theorem:

\begin{theorem} \label{th:tit}
If for given $y \in L^2_{\C}(0,2)$ the function $x \in L^2_{\C}(0,1)$ solves the equation~(\ref{eq:auto2}) with $k \equiv 1$,
then $x$ and $-x$ are the only solutions of this equation.
\end{theorem}
\begin{proof}
Let $x \in L^2_{\C}(0,1)$ and $x+\Delta \in L^2_{\C}(0,1)$ solve for all $0 \le s \le 2$ the equation~(\ref{eq:auto1}).
Then we have $[(x+\Delta)*(x+\Delta)-x*x](s)=[\Delta*(2x+\Delta)](s)=0$ for almost all $s \in [0,2]$. If we set
$$q_\Delta:=\sup\{q \ge 0:\; \Delta(\tau)=0 \;\;\mbox{for almost all}\;\; \tau \in [0,q]\},$$ then $x$ and $x+\Delta$
are different elements of $L^2_{\C}(0,1)$ if and only if $a_1:=q_\Delta<1$. For that case, Lemma~\ref{lem:tit} ensures with $a_2>1$ that
$[2x+\Delta](q)=0$ and $[x+\Delta](q)=[-x](q)$ for almost all $q \in [0,1]$. Thus $-x$ is the second solution besides $x$ and other
solutions can be excluded by this proof.
\end{proof}

It seems to be an open problem under what conditions imposed on $k \not \equiv 1$ the result of Theorem~\ref{th:tit} of having just a twofold solution of (\ref{eq:auto2}) can be formulated and proven.

\section{An adapted regularization approach} \label{sec:adapreg}

Now we return to the SD-SPIDER-motivated aspect (b), where additional data $\hat A$ are available.  To find stable approximate solutions for the operator equation (\ref{eq:opeq}) with the nonlinear operator $F:\,L^2_{\C}(0,1) \to L^2_{\C}(0,2)$ from (\ref{eq:forward}), we can exploit the nonlinear \emph{Tikhonov regularization}
(see, e.g.,~\cite[Chapter~10]{EHN96}), where the regularized solutions $x_\alpha^\delta \in L^2_{\C}(0,1)$ are minimizers of
\begin{equation}\label{eq:Tik}
\|F(x)-y^\delta\|^2_{L^2_{\C}(0,2)}+ \alpha \Omega(x) \to \min, \quad \mbox{subject to} \quad x \in D(\Omega) \subseteq L^2_{\C}(0,1),
\end{equation}
with a regularization parameter $\alpha>0$ and a \emph{stabilizing functional} $\Omega: D(\Omega) \to [0,\infty)$ (cf.~\cite[Chapter~4]{SKHK12}) having the domain $D(\Omega)$.
To prefer smooth solutions $x$ the penalty functional is frequently set as
$$\Omega(x):=\|x-\overline x\|^2_{L^2_{\C}(0,1)} \qquad \mbox{and} \qquad \Omega(x):=\|Lx\|^2_{L^2_{\C}(0,1)},\qquad \mbox{respectively},$$
where $\overline x \in L^2_{\C}(0,1)$ is a reference element and $\|x-\overline x\|^2_{L^2_{\C}(0,1)}$ attains small values if $x$ is close to $\overline x$, and on the other hand $L: D(L) \subseteq L^2_{\C}(0,1) \to L^2_{\C}(0,1)$ denotes a densely defined differential operator such that $\|Lx\|^2$ attains small values if for example first or second  derivatives of $x$ are `small'. For the calculations we chose the second derivative $Lx=\frac{\partial^2}{\partial t^2}x$.
The specific measurement situation of our inverse problem in ultrashort laser pulse characterization, where only the phase $\varphi$ in $x$ is to be determined, whereas an observation $\hat A=|\hat x|$ of the amplitude can be observed,
allows us define the \emph{problem specific rule for choosing the regularization parameter}, which in continuous formulation reads as
\begin{equation} \label{eq:rule}
\alpha_*=\alpha_*(y ^\delta,\hat A): \quad \int \limits _0^1 \left||x^\delta_{\alpha_*}(q)|- \hat A(q)\right|^2 dq \le \int \limits _0^1 \left||x^\delta_\alpha(q)|-\hat A(q)\right|^2 dq \quad \mbox{for all} \;\; \alpha>0,
\end{equation}
and to use $x^\delta_{\alpha_*}$ as the adapted approximate solution to Eq.~(\ref{eq:opeq}). Owing to that additional data information $\hat A$ the common search of $\alpha>0$ based on \emph{heuristic rules} like
the quasi-optimality rule, for example successfully applied in \cite{PHNOR01}, can be completely avoided.
Although analytic properties of (\ref{eq:rule}), especially existence of a minimizer $\alpha_*$, are hard to verify for a search over all $\alpha>0$, the situation simplifies in the practically relevant case of a set of discrete values $\{\alpha_1,\alpha_2,\dots,\alpha_N\}$ with $n\in\N$ fixed. A minimizer of (\ref{eq:rule}) over only the $\alpha_j$, $j=1,\dots,N,$ exists. In the unlikely case that $\alpha_*$ is not unique, we may pick one of the minimizers.
Since the total amount of computational work for obtaining  $x^\delta_{\alpha_*}$ from (\ref{eq:Tik}) and (\ref{eq:rule}) is rather high, \emph{iterative regularization procedures} (cf.~\cite{KNS08,BKS11} for an overview) can yield alternatives to the Tikhonov regularization with reduced computational expenses. Our focus is on a variant
of the \emph{Levenberg-Marquardt method}, which is a Newton-type method for nonlinear least-squares problems. For the mathematical theory of this method see the recent paper \cite{Hanke10}. Here, we consider
the iteration process
\begin{equation} \label{eq:LMit}
x^\delta_{(l+1)}:=x^\delta_{(l)}+\gamma\left(F^\prime(x^\delta_{(l)})^*F^\prime(x^\delta_{(l)})+\alpha L^*L \right)^{-1}F^\prime(x^\delta_{(l)})^*(y ^\delta-F(x^\delta_{(l)})),
\end{equation}
with  appropriate relaxation factors $\gamma>0$ and a regularization parameter $\alpha>0$, aimed at minimizing the linearized functional $$\|y^\delta-F(x_{(l)})-F^\prime(x_{(l)})(x-x_{(l)})\|^2_{L^2_{\C}(0,2)}+\alpha\|L(x-x_{(l)})\|^2_{L^2_{\C}(0,2)}$$ and terminated for $l=l_*$
early enough according to some stopping rule. The last iterate $x_\alpha^\delta:= x^\delta_{(l_*)}$ acts as regularized solution. In contrast to the classical version of the Levenberg-Marquardt method we do not diminish the regularization parameter $\alpha$ with growing $l \in \N$, but keep it constant. Under all such regularized solutions we select  $x^\delta_{\alpha_*}$ by the parameter choice rule (\ref{eq:rule}) and use it as approximate
solution to Eq.~(\ref{eq:opeq}). Numerical experiments proved for our problem that the squared deviations of absolute values occurring in (\ref{eq:rule}) are also helpful for any fixed $\alpha>0$ to define the index $l_*$ for stopping
the iteration process. Precisely, we always observed that the values of discretized versions of the integral $\int \limits _0^1 \left||x^\delta_{(l)}(q)|- \hat A(q)\right|^2 dq$ decrease with growing $l=1,2,...$ in the initial part of the iteration up to some turn around point with iteration number $l_*$ after which the integrals tend to increase with growing $l$. As the studies show, this turn around point is frequently connected with appropriate phase functions. This essentially motivates the stopping rule. A numerical example is given in Table ~1.

We still mention that for continuous kernels $k$ the bounded linear operators of the form $F^\prime(x_0):L^2_{\C}(0,1) \to L^2_{\C}(0,2)$ in formula (\ref{eq:LMit}) denote Fr\'echet derivative of the operator (\ref{eq:forward}) at the point
$x_0 \in L^2_{\C}(0,1)$, which can simply be verified as
$$[F^\prime(x_0) h](s)=\int \limits_0^s (k(s,q)+k(s,s-q))x_0(s-q)h(q)dq, \quad 0 \le s \le 2,\quad h \in L^2_{\C}(0,1).$$

\bigskip

Solving (\ref{eq:forward}) on a computer requires discretization of the problem, which we did as follows. The function values of $x$ are to be reconstructed at $N$ supporting points $q_n$ which are chosen equidistantly in an interval $[q_{\textup{min}},q_{\textup{max}}]$. The discrete signal is denoted by
\begin{align}\underline{x}=(x_n)_{n=1}^N=(x(q_n))_{n=1}^N=(\hat A(q_n)e^{i\varphi(q_n)})_{n=1}^N
 \end{align}
for $n=1\dots N$ and $q_n=q_{\textup{min}}+(n-1)\Delta q$ with $\Delta q=\frac{q_{\textup{max}}-q_{\textup{min}}}{N-1}$.\\
The notation for the output signal is analogous,
\begin{align}
 \underline{y}=(y_m)_{m=1}^{2N-1}=(y(s_m))_{m=1}^{2N-1}=(\hat B(s_m)e^{i\psi(s_m)})_{m=1}^{2N-1}
\end{align}
for $m=1\dots 2N-1$ and $s_m=2q_{\textup{min}}-q_\textup{cw}+(m-1)\Delta q$.\\
Here $q_\textup{cw}$ is the frequency
of the quasi-continuous wave ancilla pulse. The kernel takes the form
\begin{equation}
 \underline{K}=k_{m,n}=K(s_m,q_n,q_\textup{cw}),\hspace{1cm}m=1,2,\dots,2N-1,\hspace{0.2cm}n=1,2,\dots,N,
\end{equation}
with  $K(\cdot,\cdot,\cdot)$ from (\ref{eq:spiderkernel}). The autoconvolution operator $F$ is discretized using the rectangular rule. That way,
\begin{equation}
 y_m=\sum\limits_{j=1}^N k(s_m,q_j,q_\textup{cw})x(q_j)x(s_m+q_\textup{cw}-q_j) \Delta q,\qquad m=1,2,\dots,2N-1.
\end{equation}
Because of the finite support of $x$, $x(s_m+q_\textup{cw}-q_j)=0$ for $s_m+q_\textup{cw}-q_j<q_{\textup{min}}$
and $s_m+q_\textup{cw}-q_j>q_{\textup{max}}$.
The complete operator can be written as a multiplication of a matrix $\underline{F}(\underline{x})\in\C^{2N-1\times N}$
with the vector $\underline{x}$,
\begin{equation}\label{eq:conv_discrete}
\uly=\ulf(\ulx)\ulx
\end{equation}
where
\begin{equation}\label{eq:matrix_discrete}
\underline{F}(\ulx)=\Delta q
\begin{pmatrix}
k_{1,1}x_1 & 0 & \hdots & 0 & 0\\
k_{2,1}x_2 & k_{2,2}x_1 & \hdots & 0 & 0\\
 & \ddots & \ddots & & \vdots \\
k_{N-1,1}x_{N-1} & k_{N-1,2}x_{N-2} & \hdots & k_{N-1,N-1}x_1 & 0\\
k_{N,1}x_N & k_{N,2}x_{N-1} & \hdots & k_{N,N-1}x_2 & k_{N,N}x_1\\
0 & k_{N+1,1}x_N & \hdots & k_{N+1,N-1}x_3 & k_{N+1,N-1}x_2 \\
\vdots & &\ddots &\ddots &  \\
0 & 0 & \hdots & k_{2N-2,N-1}x_N & k_{2N-2,N}x_{N-1}\\
0 & 0 & \hdots &0 & k_{2N-1,N}x_N
\end{pmatrix}.
\end{equation}
Analogously the Fr\'{e}chet derivative $F'(x_0)$ is discretized. The m-th entry $(\underline{F}(\underline{x}_0)h)_m$ then reads as
\begin{equation}
 (\underline{F}'(\underline{x}_0)\underline{h})_m=\sum\limits_{j=0}^N (k(s_m,q_j,q_\textup{cw})+k(s_m,s_m+q_\textup{cw}-q_j,q_\textup{cw}))x_0(s_m+q_\textup{cw}-q_j)h(q_j)\Delta q.
\end{equation}
Discretizing the operator $L=\frac{\partial^2}{\partial q^2}$ as
\begin{equation}
 \underline{L}=\frac{1}{\Delta q^2}\begin{pmatrix}
    2  & -1 & 0 & 0 &\hdots& 0\\
    -1 & 2 & -1 & 0 &\hdots& 0\\
    0 &  \ddots&\ddots & \ddots &  & & \\
   0& &\hdots&  & -1 & 2\\
   \end{pmatrix},
\end{equation}
the iteration rule (\ref{eq:LMit}) reads as
\begin{equation} \label{eq:LMit_discrete}
\ulx^\delta_{(l+1)}:=\ulx^\delta_{(l)}+\gamma\left(\ulf^\prime(\ulx^\delta_{(l)})^*\ulf^\prime(\ulx^\delta_{(l)})+\alpha \ull^*\ull \right)^{-1}\ulf^\prime(\ulx^\delta_{(l)})^*(\uly ^\delta-\ulf(\ulx^\delta_{(l)})).
\end{equation}
As initial value $\ulx^\delta_{(0)}$ we take the measured absolute values and a zero phase.
\bigskip

\section{Case studies}\label{sec:casestudies}
Before turning to examples for the reconstruction of a phase, we want to give an idea about the dependency of the measurements on the unknown phase. An example is given in Figure~2. 
There, two slightly different phases are shown together with their autoconvolution signal split into phase and absolute values. Both convolutions have been performed using the same function for the absolute values $\hat A$ and the physical kernel. Although the resulting phases look quite alike, there is a much more obvious difference in the resulting absolute values $\hat B$. The example indicates that the complete autoconvolution signal is necessary for the reconstruction.
\begin{figure}
\centering
\includegraphics[width=\linewidth]{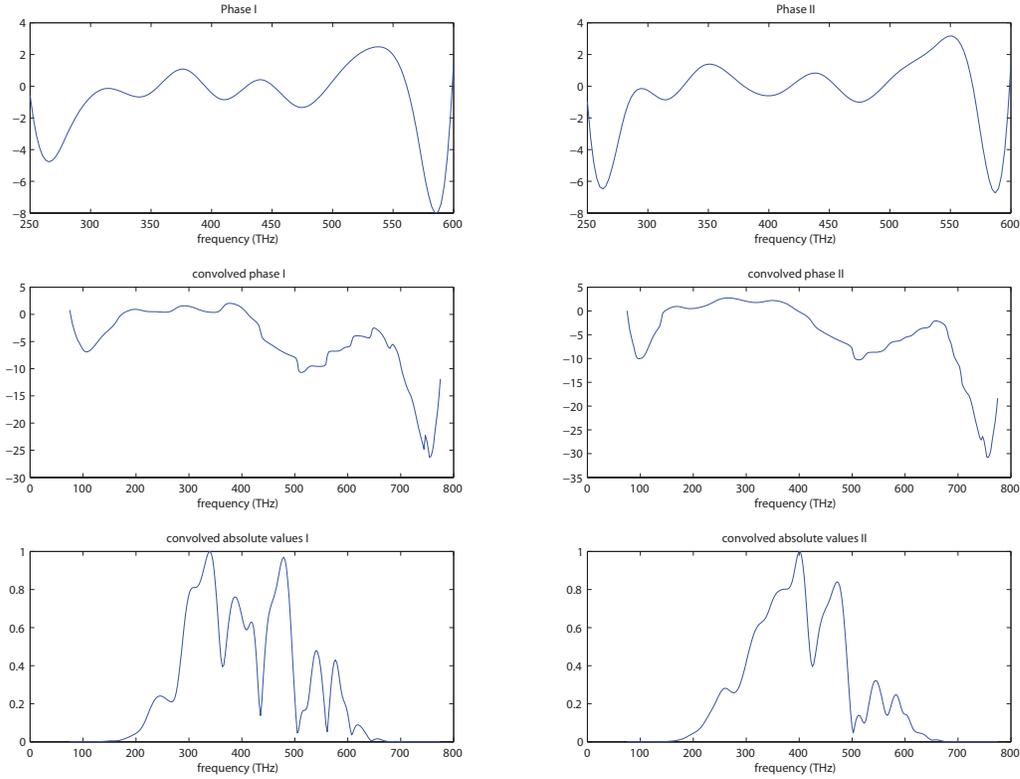}\caption{In the top row, two different, yet similar looking examples of phases for the unknown $x$ are shown. Both are convolved with the same function for the absolute values. In the second row, the phases of the convolved pulse $y$ are plotted, again being quite alike. However, the absolute values of $y$, pictured in the bottom row, show larger distinctions.}
\end{figure}\label{fig:phasenvergleich}
In the following we give two examples for the reconstruction. For a more detailed discussion about the numerical results and effects we refer to \cite{Gerth11}. To test our reconstruction algorithm, we chose functions for the absolute values $\hat A(q)$ and for the phase $\varphi(q)$. Since the real magnitude of $\hat A(q)$ is unknown, we did most of the testing with a maximum value of $\hat A_{\rm max}=10^{-7}$ so that both $x$ and $y$ were neither too large nor too small. Since the kernel adds a magnitude of $10^{28}$ we arrived at approximately $10^{14}$ for the SD-Spider signal. To get reasonable values for the regularization parameter, we rescaled it. With $\alpha$ and $\Delta q$ from the previous section we define $\hat{\alpha}=\alpha\cdot \hat A_{\rm max}^2 \cdot \Delta q^{-4} \cdot (10^{28})^{-2}$. The normalized parameter $\hat{\alpha}$ will be the one given in the evaluation of the experiments. In order to follow the physical background of the problem, huge numbers appear inevitable. Both pulses
were sampled on a fine grid and according to (\ref{eq:conv_discrete}). From this representation we extracted absolute values $\hat B(s)$ and phase $\psi(s)$, which we again sampled on a coarser grid to avoid inverse crime, i.e., an overly good reconstruction that arises as an artifact because forward and backward calculations were either performed on the same grid or because one grid was a multiple of the other. To all data that is assumed to be given as measurements, i.e. the absolute values of both the spectrum of the pulse to be reconstructed $\hat A(q)$ and of SD-Spider signal $\hat B(s)$, as well as the phase of the SD-signal $\psi(s)$, we added to each data point normally distributed random noise with zero mean and a standard deviation of $\delta$ percent relative to the correct value. The final regularization parameter was chosen according to (\ref{eq:rule}), using the noisy data $\hat A$ as reference. This means that we assume the solution to be best if the deviation of the corresponding absolute values is minimal.

\begin{figure}
\centering
\subfigure[Fundamental pulse]{\includegraphics[width=0.9\linewidth]{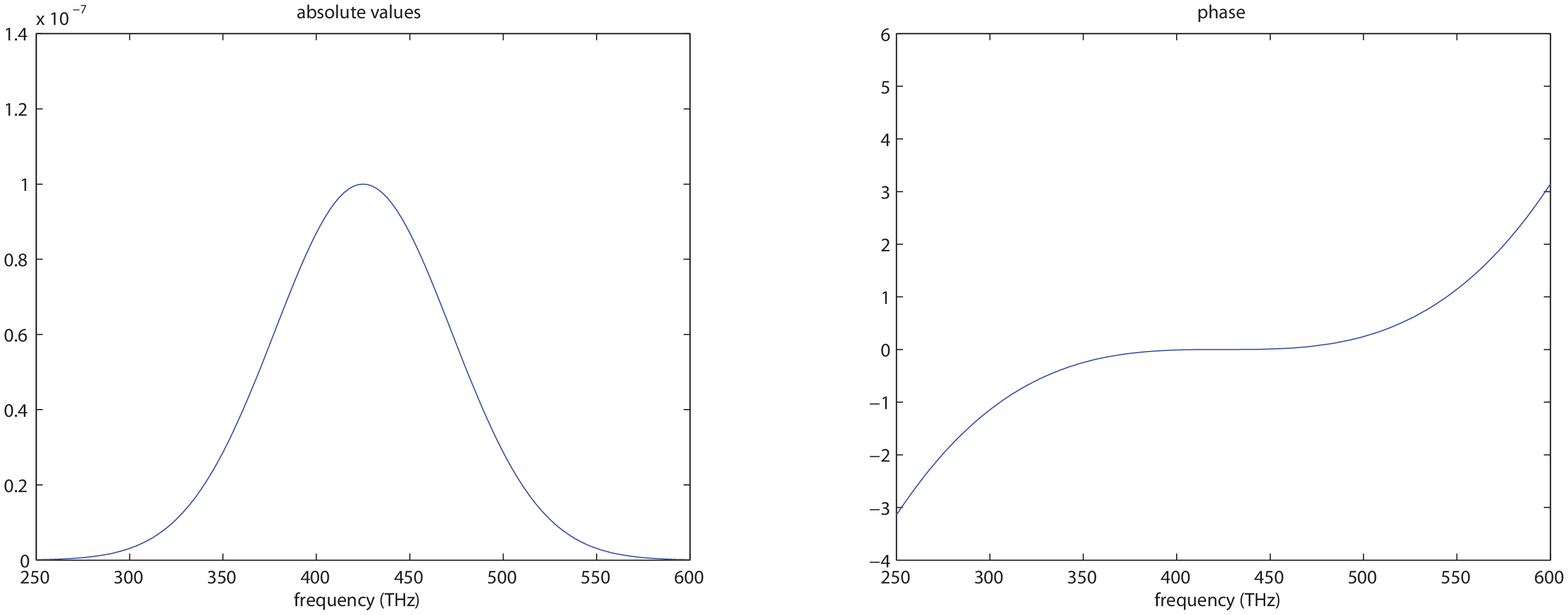}}
\subfigure[Result of the autoconvolution of the fundamental pulse, where we added $5\%$ relative noise]{\includegraphics[width=0.9\linewidth]{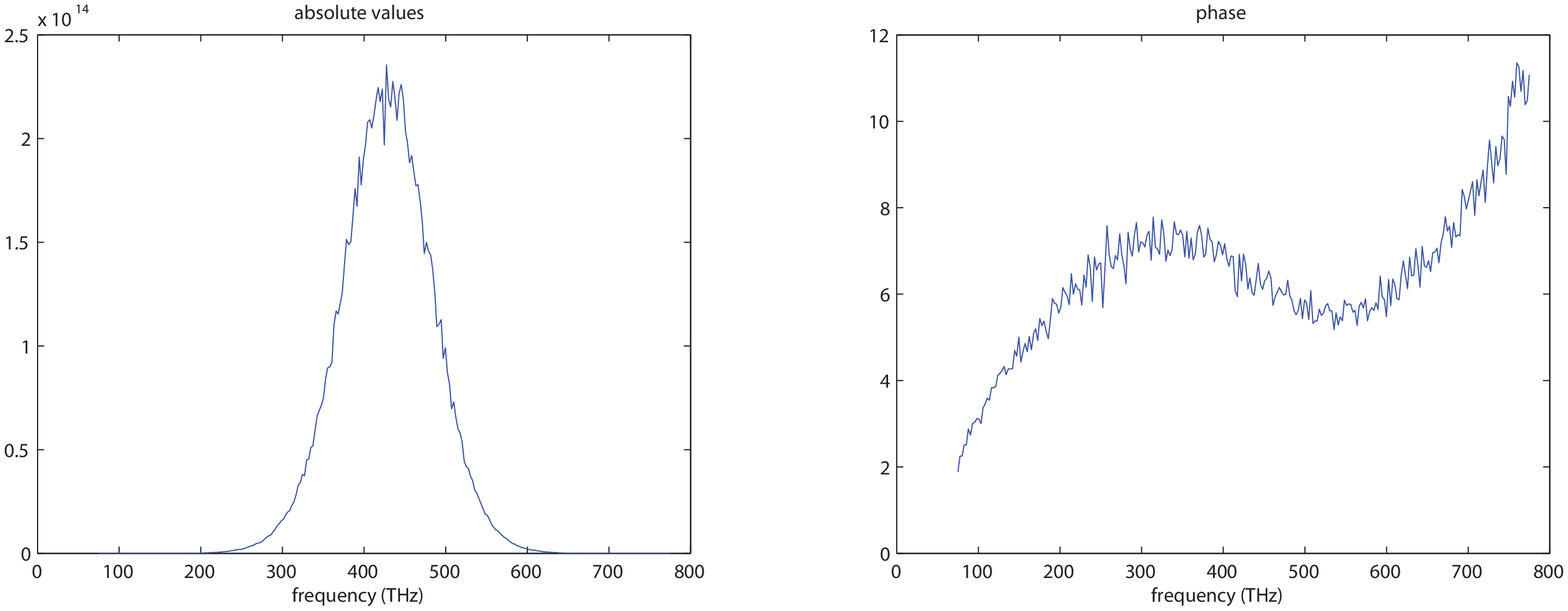}}
\subfigure[Reconstruction of the fundamental pulse and comparison with the true solution, $\delta=5\%$, $\hat{\alpha}=5.86\cdot10^{6}$]{\includegraphics[width=0.9\linewidth]{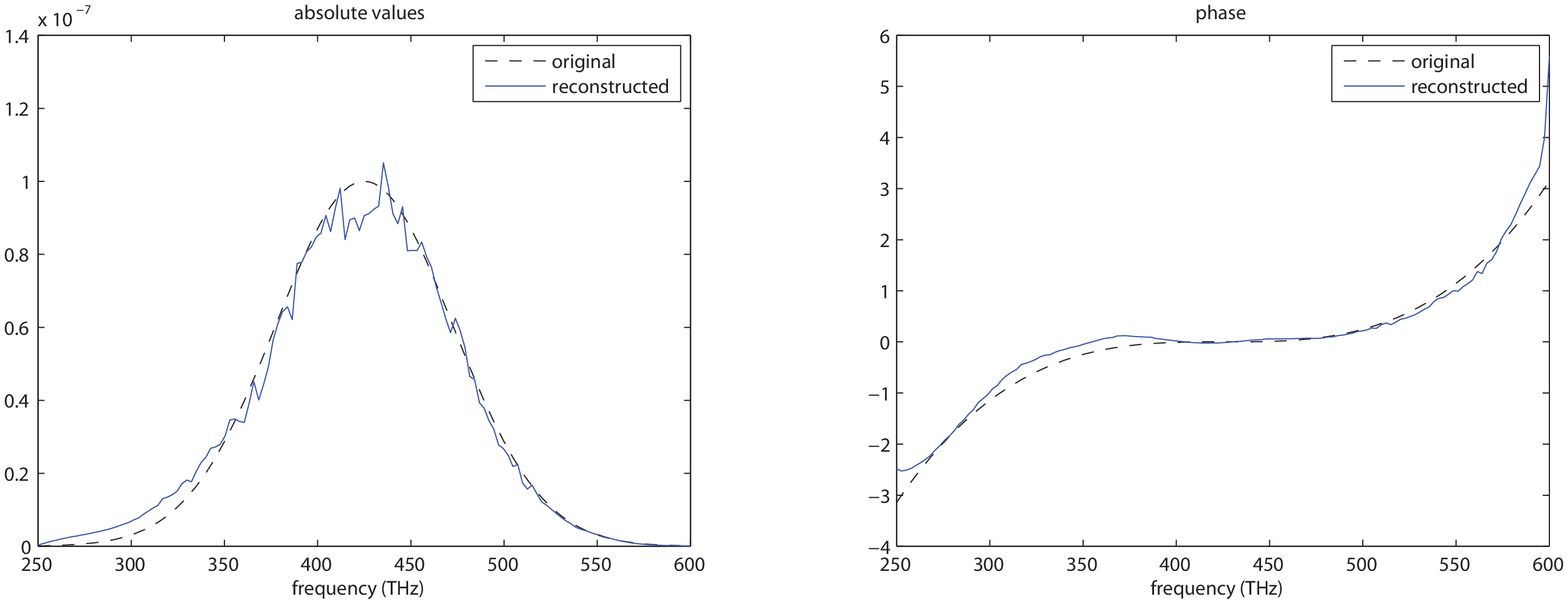}}\caption{First example using a very smooth fundamental pulse}
\end{figure}\label{fig:pulse_smooth}

\begin{figure}
\centering
\subfigure[Fundamental pulse]{\includegraphics[width=0.9\linewidth]{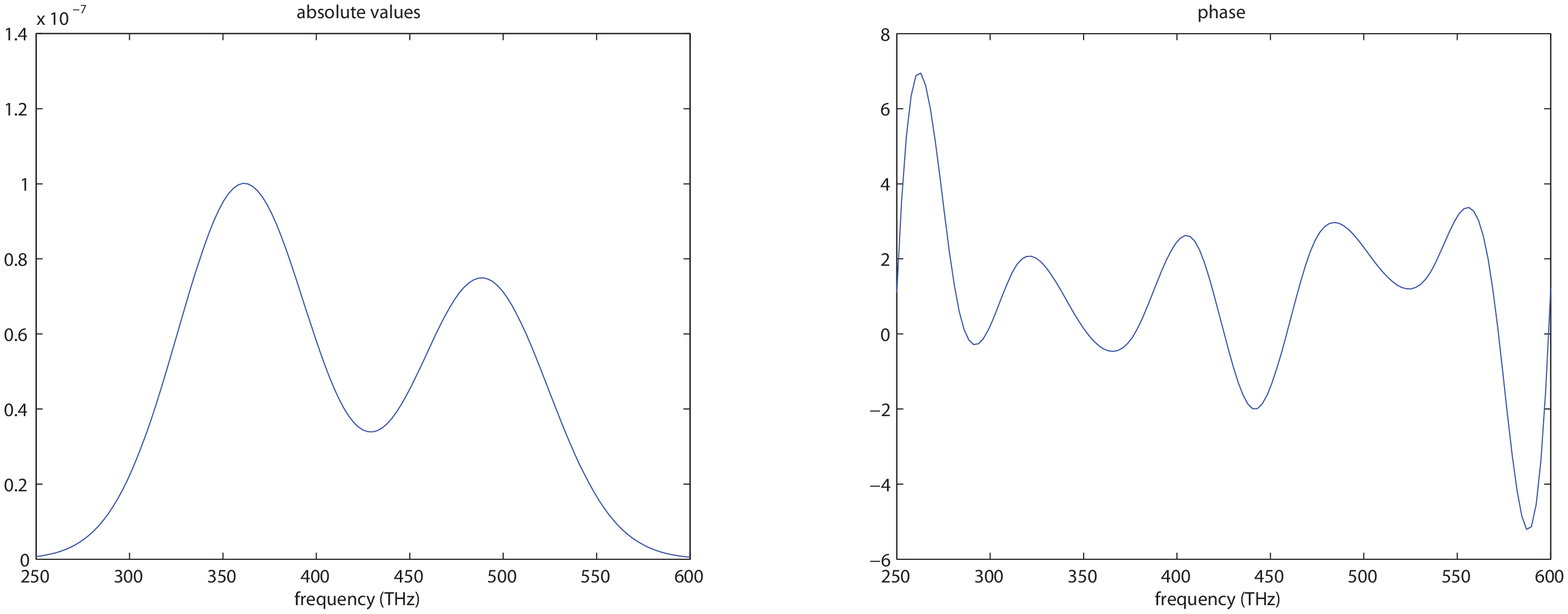}}
\subfigure[Result of the autoconvolution of the fundamental pulse]{\includegraphics[width=0.9\linewidth]{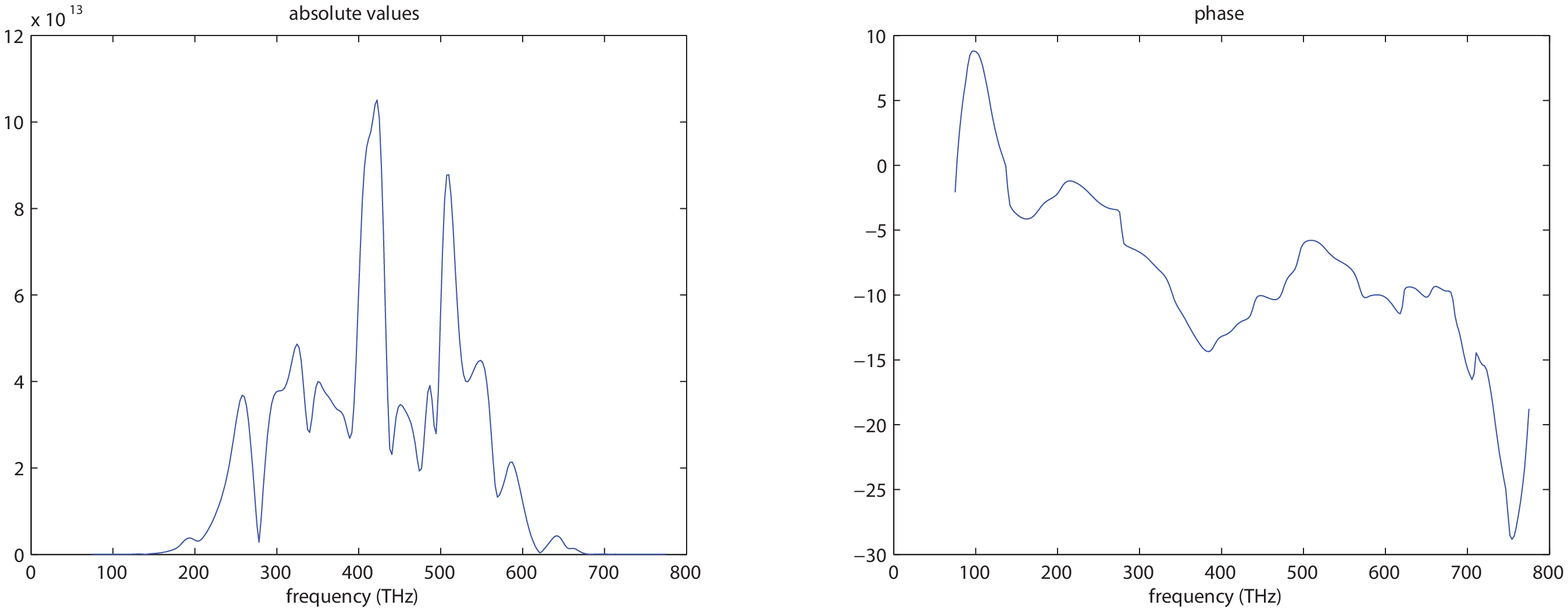}}
\subfigure[Reconstruction of the fundamental pulse and comparison with the true solution, noise-free, $\hat{\alpha}=2.17$]{\includegraphics[width=0.9\linewidth]{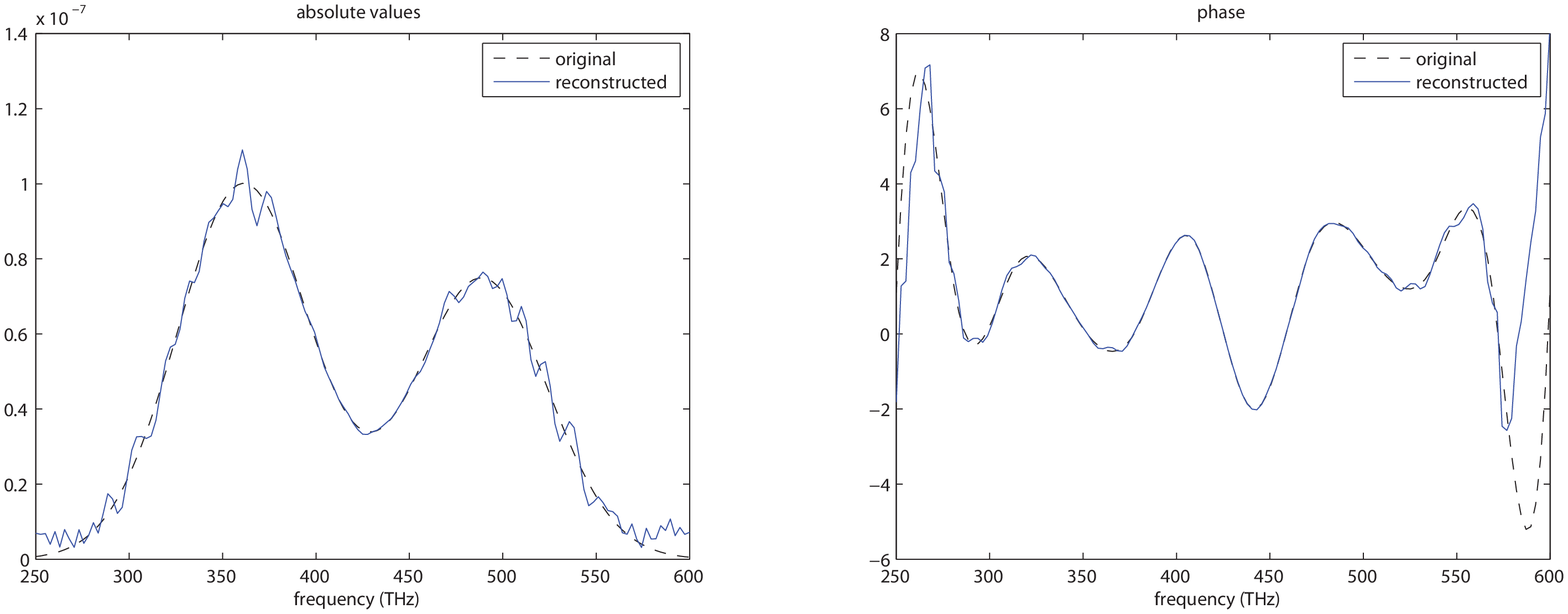}}
\subfigure[Reconstruction of the fundamental pulse and comparison with the true solution with noise $\delta=1\%$, $\hat{\alpha}=2.82$]{\includegraphics[width=0.9\linewidth]{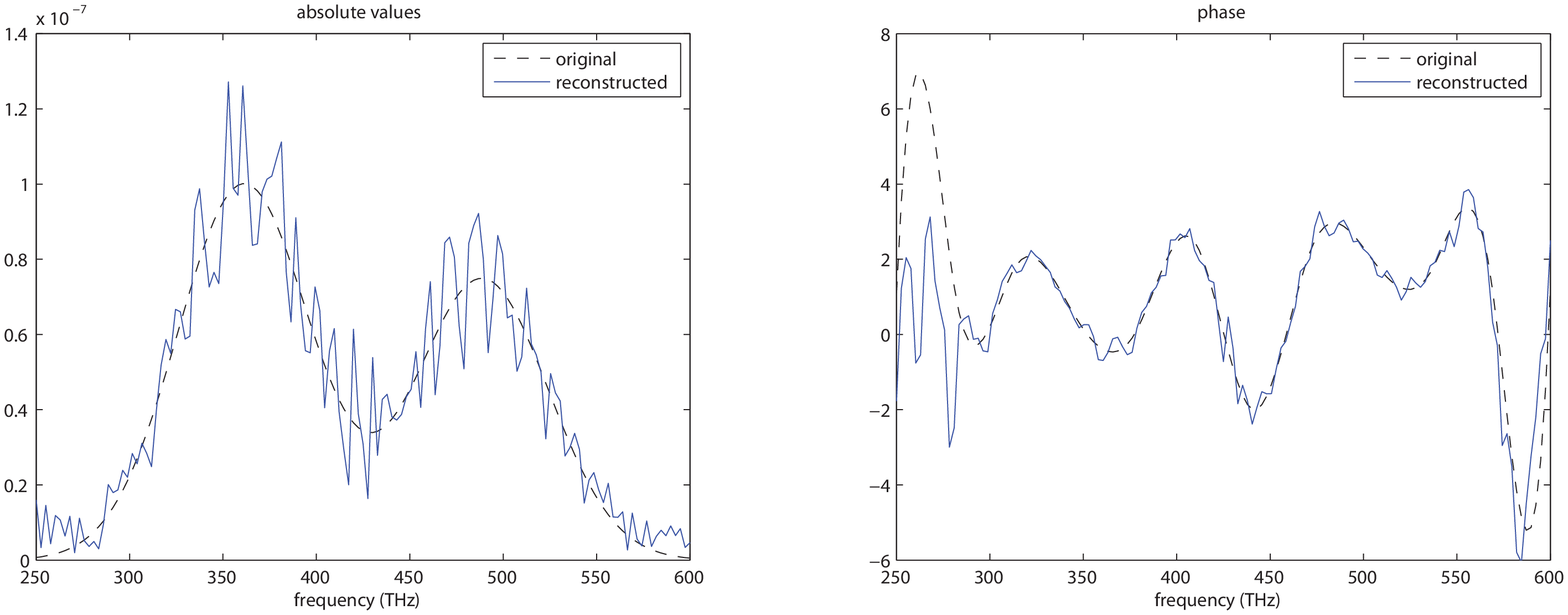}}\caption{Second example with a more complicated fundamental pulse. See Table~1 for characteristic values of certain iterations.}
\end{figure}\label{fig:pulse_complex}

\begin{table}\label{table:iterations}
\begin{tabular}{|c|c|c|c|}
\hline
Iteration & Residual & Solution smoothness & Deviation of absolute values\\
${\scriptscriptstyle(l)}$ &  ${\scriptscriptstyle ||\underline F(\underline x^\delta_{(l)})-\underline y^\delta||}$ &   ${\scriptscriptstyle ||\underline L x^\delta_{(l)} ||}$ &   ${\scriptscriptstyle ||\hspace{3pt} |\underline x^\delta_{(l)}|-\hat {\underline A}||}$  \\\hline
1 & 9.5819e-01 & 1.6806 & 0.5252\\
20 & 2.4115e-02 & 4.3935 & 0.7916 \\
40 & 2.0682e-02 & 4.9909 & 0.7937 \\
60 & 1.5369e-02 & 4.6496 & 0.6077 \\
80 & 2.2482e-03 & 3.9273 & 0.2563 \\
100 & 1.3792e-03 & 3.7857 & 0.1964 \\
120 & 1.1022e-03 & 3.9429 & 0.1701 \\
140 & 9.4595e-04 & 4.3199 & 0.1623 \\
142 & 9.3083e-04 & 4.3642 & 0.1622 \\
143 & 9.2340e-04 & 4.3865 & 0.1622 \\
144 & 9.1606e-04 & 4.4090 & 0.1623 \\
150 & 8.7480e-04 & 4.5451 & 0.1632 \\
200 & 5.7859e-04 & 5.8612 & 0.1832 \\
250 & 3.1613e-04 & 7.2540 & 0.2020\\\hline
\end{tabular}\caption{Characteristic values of the iteration process for the solution in Fig.~3(d).\\ 
All values have been normalized to simplify comparison. Since the iteration is aimed at minimizing the residuals, the values of the residuals decrease steadily during the process. The iteration must be terminated early enough,
otherwise the solution smoothness gets lost for late iterates. Moreover, the solutions increasingly oscillate as a consequence of the ill-posedness. On the other hand, the deviation of the absolute values decreases after some
 starting phase until they reach some turning point. At iteration $l_*=143$ the minimum is reached, and according to our stopping rule we choose this as the final solution. If the iteration is continued the difference in absolute values increases again, and the reconstructions tend to the worse.}
\end{table}

\subsection{Case I: A very smooth pulse}
As a first example we consider a very smooth pulse shown in Figure~3(a). 
The absolute values possess only one peak and the values of the phase increase from $-\pi$ and $\pi$. It should be mentioned that both functions the phase and its first derivative are zero in the middle of the frequency domain. Figure~3(b)
shows the result of the autoconvolution including the added noise. The noise level here is $\delta=5\%$. The best reconstruction was achieved for $\hat{\alpha}=5.86\cdot10^{6}$. Since we are looking for a very smooth solution and use a smoothness penalty for the regularization, the parameter is very large. The result is given in Figure~3(c). 
We were able to recover the phase in an acceptable way. Although there are problems for lowest and highest frequencies, the middle part is retrieved close to the original phase.	Due to the fairly large value of the regularization parameter, only few remainders of the highly oscillating noise (cf. Figure~3(b)) 
are still visible.

\subsection{Case II: An oscillating phase}
As a second example we have a more complex situation in mind. Here we choose an amplitude function with two peaks leading to the noisy function $\hat A(q)$. The phase function shown in Figure~4(a) 
has more of a sinusoidal structure.
Figure~4(b) again shows the result after the autoconvolution, this time without any noise. Even in this case we need an appropriate parameter $\hat{\alpha}$ to recover the phase in an acceptable way. The noise-free reconstruction is shown in Figure~4(c) with $\hat{\alpha}=2.17$. While in the middle part both phases coincide nearly perfectly, problems arise at the boundaries as we already observed in the previous example. Because of the structure of the autoconvolution equation, the amount of information on the boundaries is much lower than in the middle, c.f. (\ref{eq:conv_discrete}) and (\ref{eq:matrix_discrete}). The higher the noise level $\delta$, the more severe those problems become. In the last figure a reconstruction for $\delta=1\%$ and  $\hat{\alpha}=2.82$ is shown where a lot of noise artifacts remained. However, it seems that the absolute values are influenced more by the measurement errors than the phase. Numerical values for certain iterations are shown in Table ~1. Especially our criterion to
stop the iteration for fixed $\alpha$ is clearly visible there. Namely, at $l_*=143$, the difference of the absolute values of the iterate to the noisy reference values reaches its minimum. The solution at this iteration is the one shown in Figure ~4(d).

\bigskip

\section{Conclusions}
In this paper, we have studied a new type of kernel-based autoconvolution problems, for which a stable approximate solution is required for measurements of ultrashort laser pulses with the self-diffraction SPIDER method. The problem is formulated as a nonlinear integral equation with complex-valued functions over a finite real interval on both sides of the equation. With respect to the mathematical model, the novelty of this inverse problem consists in the occurrence of a physically motivated complex kernel function and the observability of the amplitude of the incident pulse such that the focus of the paper is on phase reconstruction.  After reviewing recent developments in ultrashort pulse characterization we have introduced an appropriate physical model (\ref{eq:spiderfreqsimple}) with  kernel (\ref{eq:spiderkernel}) and its mathematical description as a nonlinear operator equation. Using this abstract setting, which supports the mathematical analysis of the problem, we outlined by
Example~\ref{ex:localill} the local ill-posedness of the inverse problem under consideration. That is, for given data the solution cannot be approximated arbitrarily precise, even if the noise on the data is sufficiently small. Titchmarsh's convolution theorem formulated for our model in Lemma~\ref{lem:tit} enabled us to show that the autoconvolution equation has two complex conjugate solutions in case of a trivial kernel. However, the situation for arbitrary kernels seems to be an open question. A  main goal of this article was to suggest and test an adapted regularization approach. We used a variant of the Levenberg-Marquardt algorithm (\ref{eq:LMit}) in the discretized form (\ref{eq:LMit_discrete}). A crucial point in the regularization process is the choice of the regularization parameter $\alpha$. Since the amplitude function can be measured as part of the solution, we could motivate a specific parameter choice rule (\ref{eq:rule}). Under all regularized solutions, which are calculated for varying $\alpha>0$, we recommend to take the one that approximates the observed absolute values in an optimal way. It is known that iterates for regularized solutions to ill-posed problems tend to worsen whenever the iteration is not stopped early enough. To find an appropriate stopping rule, we monitored the deviation between the measured and calculated absolute values, which reaches a minimum during the iteration process and then increases again. This minimum is used to stop the iteration for fixed $\alpha$. Table~1 illustrates such behavior for some numerical example. Case studies with two different pulses and different noise levels indicate opportunities and limitations of our regularization method.

From a physicist's point of view, the formalism described here opens a perspective for using a wider class of nonlinearities for the SPIDER pulse characterization method, making it more universally applicable. $\chi^{(3)}$ nonlinearities such as the self-diffraction process are much more widespread than  $\chi^{(2)}$ processes that have been used nearly exclusively for SPIDER so far. With the solution of the $\chi^{(3)}$ SPIDER autoconvolution problem at hand, the integration of a complete optical pulse characterization set-up in an integrated optical device becomes feasible. As SPIDER is the only method that does not require mechanical variation of an optical delay this enables, in principle, the construction of an optical oscilloscope on chip, of course with the requirement to externally process the measured data by the discussed regularization procedure. Moreover, there appear other interesting applications of $\chi^{(3)}$ SPIDER variants for the characterization of broadband ultraviolet pulses, where $\chi^{(2)}$ processes do not offer viable options. For all these intriguing applications, we now demonstrated a viable way of retrieving the relevant phase data by regularization of the respective autoconvolution problem.

\section*{Acknowledgements}
We thank four referees for reading the paper carefully. The given hints allowed us to state the aim of the paper more precisely.
B.~Hofmann was partly supported by Grant 1454/8-1 from
Deutsche Forschungsgemeinschaft (DFG). D.~Gerth was funded by the Austrian Science Fund (FWF): W1214-N15, project DK08. G.~Steinmeyer gratefully
acknowledges support by the Academy of Finland (Project Grant No. 128844).


\label{lastpage}

\end{document}